%% file: main.tex
\documentclass{ieeeaccess}

\usepackage{cite}
\usepackage{amsmath,amssymb,amsfonts}
\usepackage{booktabs}
\usepackage{graphicx}
\usepackage{array}
\usepackage{multirow}
\usepackage{tabularx}
\usepackage{longtable}
\usepackage{makecell}
\usepackage{placeins}
\usepackage{microtype}
\usepackage{balance}
\usepackage{textcomp}
\usepackage{url}
\usepackage[hidelinks]{hyperref}
\graphicspath{{figures/}}
\input{data/final_numbers.tex}

\makeatletter
\providecommand{\xfigwd}{\columnwidth}
\makeatother

\newcommand{\JSPT}{J_{\mathrm{SPT}}}
\newcommand{\JH}{J_{\mathrm{H}}}
\newcommand{\JW}{J_{\mathrm{W}}}
\newtheorem{proposition}{Proposition}

\newcounter{algorithm}
\newenvironment{algorithm}[1][tbp]{\begin{figure}[#1]\refstepcounter{algorithm}}{\end{figure}}
\newcommand{\algorithmcaption}[1]{\par\smallskip\footnotesize\textbf{Algorithm \thealgorithm.} #1}

\begin{document}
\history{Date of publication xxxx 00, 0000, date of current version xxxx 00, 0000.}
\doi{10.1109/TQE.2020.DOI}

\title{Plateau-Constrained Selection: Exploiting Degeneracy for Lower-Depth Quantum Compilation}
\author{\uppercase{Owen Friedewald}\authorrefmark{1},
\uppercase{Ali Shiri Sichani}\authorrefmark{1}, and
\uppercase{Chi-Ren Shyu}\authorrefmark{1}}
\address[1]{Department of Electrical Engineering and Computer Science, University of Missouri, Columbia, MO 65211 USA (e-mail: omfvq4@missouri.edu; asp9f@missouri.edu; shyuc@missouri.edu). ORCID: 0009-0003-6847-1706 (Friedewald), 0000-0002-7239-1341 (Shiri Sichani), 0000-0001-9197-9522 (Shyu)}
\tfootnote{This article extends the authors' four-page paper from the Quantum Reinforcement Learning Workshop at the 2026 IEEE International Conference on Quantum Computing and Engineering (QCE), publicly available as arXiv:2607.15307. The journal article adds formal plateau analysis, a two-stage selection method, independent routing controls, and separate hardware studies; no figure or table from that paper is reproduced. This work received no specific financial support or grant.}
\markboth
{Friedewald \headeretal: Plateau-Constrained Selection for Lower-Depth Quantum Compilation}
{Friedewald \headeretal: Plateau-Constrained Selection for Lower-Depth Quantum Compilation}
\corresp{Corresponding author: Ali Shiri Sichani (e-mail: asp9f@missouri.edu).}

\begin{abstract}
Minimum-cost orders of commuting phase terms can produce substantially different routed circuits. On the same 36-term instances, three orders with identical support cost 74 yield mean routed depths of 228.6, 233.8, and 256.7. We exploit this degeneracy under fixed placement and maintained-parity lowering: Stage~1 attains the support optimum, and Stage~2 selects minimum routed depth among 24 equal-cost orders. For distinct pair supports, we characterize orders attaining the support lower bound through Hamiltonian paths of the support line graph and count optima exactly through 20 terms. On synthetic 16-qubit assignment-Ising instances, selection reduces depth by 12.83\% under a different SABRE routing seed, with lower depth in all 20 instances. The selected orders lie a median 1.57 pool standard deviations below the pool mean, consistent with ordinary best-of-24 selection; the useful feature is that candidate rankings persist across routing seeds. Depth reductions extend to 48 terms and a random-MaxCut generator, whereas evaluation with BasicSwap reverses the gain. A 40-instance IBM Heron study measures a 0.59\% error reduction on the executed stabilizer-probe panel, but the primary confidence interval across instances includes zero. Plateau selection therefore improves routed depth in the tested SABRE pipeline while preserving the logical support optimum.
\end{abstract}
\titlepgskip=-15pt
\maketitle

\section{Introduction}
\label{sec:introduction}

\PARstart{A}{ quantum} compiler can minimize one cost while leaving an important choice to its tie-breaking rule. For commuting phase terms, different execution orders can require the same number of logical parity changes yet produce circuits with different depths after routing. Choosing among those orders provides a further opportunity to reduce compilation cost. This paper asks whether those ties can be used deliberately. We first minimize the support-transition cost of a phase component, then search within its minimum-cost set for an order that routes to lower depth. We call this set a \emph{plateau}: membership means equal primary cost, without implying that every member is reachable from every other by our local moves. The second stage preserves the primary optimum exactly.

We hold logical-to-physical placement and the maintained-parity implementation fixed, and change only term order. An auxiliary qubit (the ancilla), initialized to $\lvert0\rangle$, holds the parity needed by each phase rotation, reuses shared contributions between terms, and returns to $\lvert0\rangle$ at the end. This isolates an ordering decision that a compiler can make after placement while preserving the component's logical action. Our experiments evaluate this decision on synthetic phase components. On the same 36-term instances, three ways of choosing a minimum-support-cost order gave mean routed depths of 228.6, 233.8, and 256.7. All three had support cost 74. Thus, attaining the primary optimum did not determine the routed result. Selecting by depth among 24 equal-cost candidates reduced the first of these means to 199.3 when evaluated with a different routing seed. Figure~\ref{fig:workflow} shows the method and the separation between selection and evaluation.

We develop this argument in three steps. First, we characterize the minimum-cost set for distinct pair supports through Hamiltonian paths of the support line graph, and count its members exactly on small instances. Second, we test whether choosing among these orders improves depth beyond arbitrary candidate choice or an equal routing budget spent on transpiler restarts. Third, we test the limits of that choice across problem size, a second generator, a different router, and quantum hardware. Selection improves depth across SABRE routing seeds on the assignment and MaxCut instances. Evaluation with BasicSwap reverses the direction, showing that the choice depends on the router. The hardware study then tests whether lower depth translates into lower execution error. Together, these experiments connect the ordering freedom to its practical benefit and identify where that benefit ends.

This article extends our previous study~\cite{friedewald2026workshop}, which used a per-instance stochastic ordering search. The additions are the plateau characterization, strict secondary selection, independent routing controls, and separate hardware evaluations. The companion Supplementary Material (SM) documents the earlier objectives and hardware studies separately from the new selector.

\begin{figure*}[t]
\centering
\includegraphics[width=\textwidth]{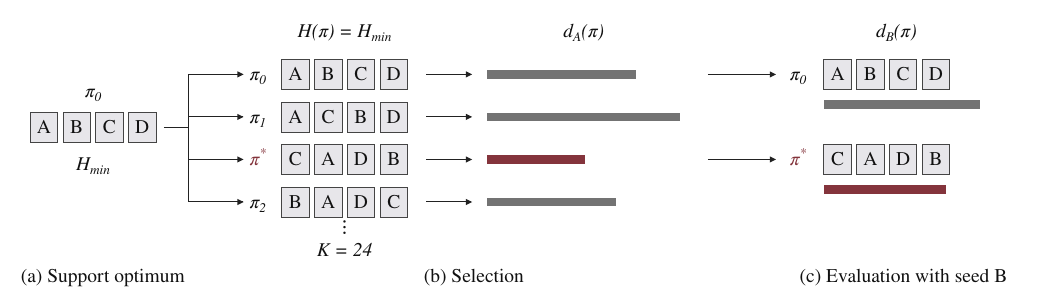}
\caption{Plateau-constrained selection. (a) The reference order $\pi_0$ attains minimum closed support cost $H_{\min}$ and initializes the search without routing results. (b) Among $K=24$ equal-cost orders (distinct up to reversal), $\pi^*$ minimizes routed depth $d_A$ under SABRE seed~A. Accent-colored labels and bars identify $\pi^*$. (c) Both orders are evaluated with seed~B on the same instance, without selecting again; seed roles are then reversed and results averaged. Phase terms, placement, and maintained-parity lowering remain fixed. Letters identify terms. Four candidates are shown schematically; bars share an illustrative scale and do not represent measured depths.}
\label{fig:workflow}
\end{figure*}

\section{Ordering Freedom and Its Primary Cost}
\label{sec:contract}

\subsection{Commuting terms and parity reuse}

Let $\mathcal P=\{(S_i,\theta_i,p_i)\}_{i=1}^n$ denote term occurrences: $S_i$ is the set of logical qubits acted on, $\theta_i$ is its phase coefficient, and $p_i$ records its source identity. Duplicate supports remain distinct occurrences. With $Z_S=\bigotimes_{q\in S}Z_q$ and identity elsewhere, commutation gives
\begin{equation}
\prod_{k=1}^{n}e^{-i\theta_{O_k}Z_{S_{O_k}}/2}
=\prod_{i=1}^{n}e^{-i\theta_iZ_{S_i}/2}
\label{eq:commutation}
\end{equation}
for every complete permutation $O$. The compiler preserves each occurrence, coefficient, and source identity, along with the placement and the deterministic mapped realization of each support. Every output must contain all occurrences exactly once, apply one phase per occurrence, and return the ancilla to $\lvert0\rangle$. We checked these conditions before routing every ordering circuit and rejected candidates that failed them.

Figure~\ref{fig:maintained} illustrates why order affects the logical cost. Moving from support $\{q_0,q_1\}$ to $\{q_0,q_2\}$ requires removing $q_1$ from the stored parity and adding $q_2$. The shared $q_0$ contribution remains. Separately computing and uncomputing each phase would lose that reuse. The logical correctness argument follows the parity stored in the ancilla. On a computational-basis data state $\lvert x\rangle$, the ancilla before term $i$ stores $a_i=\bigoplus_{q\in S_i}x_q$. An $R_z(\theta_i)$ rotation contributes $\exp[-i\theta_i(-1)^{a_i}/2]$, the phase of $e^{-i\theta_i Z_{S_i}/2}$. Toggling the qubits in $S_i\triangle S_j$ changes $a_i$ to $a_j$, so shared contributions remain available for the next term. Final uncomputation returns the ancilla to $\lvert0\rangle$ without changing the accumulated data phase. This establishes the logical construction on basis states and hence, by linearity, on arbitrary inputs. Routing must preserve that circuit's action; the ordering comparison holds its placement and lowering rules fixed.

\begin{figure*}[t]
\centering
\includegraphics[width=0.88\textwidth]{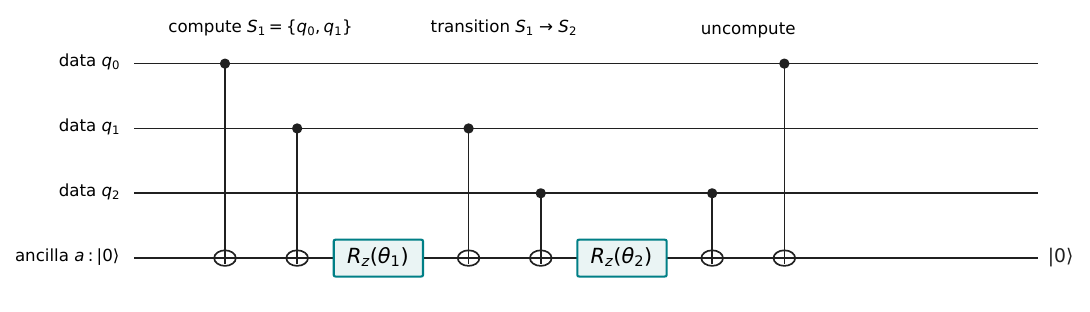}
\caption{Parity reuse between two commuting phase terms. The ancilla changes from storing the parity of $\{q_0,q_1\}$ to that of $\{q_0,q_2\}$ without an intermediate reset, then returns to $\lvert0\rangle$. This logical construction uses six CNOTs, compared with eight for independent compute--phase--uncompute sequences. Physical routing can add further operations.}
\label{fig:maintained}
\end{figure*}

\subsection{Minimum support cost and degeneracy}

Represent support $S_i$ by a binary mask $m_i$, and write $d_H(i,j)=\operatorname{popcount}(m_i\oplus m_j)$. The closed support cost is
\begin{equation}
\begin{gathered}
H_{\rm closed}(O)=\sum_{k=0}^{n}d_H(O_k,O_{k+1}),\\
O_0=O_{n+1}=0,\qquad m_0=0.
\end{gathered}
\label{eq:hclosed}
\end{equation}
The empty endpoints account for initial computation and final uncomputation. This is a logical transition cost; it does not include the physical router's SWAP and scheduling decisions. For distinct pair supports, consecutive terms either share one qubit or are disjoint. Their transition costs are therefore two or four, while the two empty endpoints contribute four in total. If $k$ consecutive pairs are disjoint,
\begin{equation}
H_{\rm closed}(O)=2n+2+2k,\qquad k\in\{0,\ldots,n-1\}.
\label{eq:strata}
\end{equation}
There are at most $n$ cost values for $n!$ orders. This coarse cost explains the possibility of ties, but does not establish how many orders attain the minimum.

\begin{proposition}[Attainment and multiplicity]
\label{prop:linegraph}
Let $G$ have logical qubits as vertices and distinct pair supports as edges. Let $L(G)$ be its line graph, whose vertices are the edges of $G$, with adjacency when two edges share an endpoint. Then $H_{\rm closed}=2n+2$ is attainable exactly when $L(G)$ has a Hamiltonian path. When attainable, minimum-cost orders are exactly the oriented Hamiltonian paths of $L(G)$; identifying an order with its reversal gives the number of undirected Hamiltonian paths.
\end{proposition}
\begin{IEEEproof}
Equation~\ref{eq:strata} attains its lower bound precisely when every consecutive pair of supports shares a qubit. That is precisely adjacency of their vertices in $L(G)$. Visiting each occurrence once is therefore a Hamiltonian path, and reversing the order reverses its orientation.
\end{IEEEproof}

An order with cost $2n+2$ therefore certifies the support optimum. We attained this bound in 11 of 18 exact-counting instances and in every tested 36- and 48-term instance. Finding such an order is a separate computational problem: edge-Hamiltonian-path decision is NP-complete on unrestricted graphs~\cite{bertossi1981edge}, and even a connected line graph can lack a Hamiltonian path.

\subsection{Multiplicity of minimum-cost orders}

We counted minimum-cost orders using Held--Karp dynamic programming augmented with exact path counts~\cite{heldkarp1962}. For a visited occurrence set $A$ ending at $j$, let $c(A,j)$ be the minimum cost from the empty depot and $q(A,j)$ its number of minimizers. Initialize $c(\{j\},j)=d_H(0,j)$ and $q(\{j\},j)=1$. For larger sets,
\begin{equation}
\begin{aligned}
c(A,j)&=\min_{i\in A\setminus\{j\}}
 [c(A\setminus\{j\},i)+d_H(i,j)],\\
q(A,j)&=\sum_{i\in I(A,j)}q(A\setminus\{j\},i),
\end{aligned}
\label{eq:counting}
\end{equation}
where $I(A,j)$ contains every minimizing predecessor. For the full occurrence set $U$, close each path with $d_H(j,0)$, take the minimum closed cost, and sum $q(U,j)$ over endpoints attaining it. Summing all minimizing predecessors counts every optimal path. The implementation uses $O(n^2 2^n)$ time and $O(n2^n)$ stored states; all reported counts fit exact unsigned 64-bit arithmetic.

Table~\ref{tab:plateaucounts} reports three instances chosen in advance at each size. Counts include both reversal orientations and vary substantially between instances. Although the absolute counts can be large, the fraction of all orders that is optimal decreases over this size range. The search at 36 and 48 terms complements these exact counts by testing whether enough distinct optima can be found for secondary selection.

\begin{table}[t]
\caption{Exact number $W$ of minimum-cost orders for three instances per size. When Proposition~\ref{prop:linegraph}'s lower bound is attained, half of $W$ is the number of undirected Hamiltonian paths of the support line graph. Counts include both reversal orientations. ``Middle'' denotes the middle of the three instance counts; variation between instances accounts for the non-monotone sequence.}
\label{tab:plateaucounts}
\centering\footnotesize
\begin{tabular}{rrrr}
\toprule
$n$ & min $W$ & middle $W$ & max $W$ \\
\midrule
10 & 24 & 288 & 288 \\
12 & 96 & 96 & 936 \\
14 & 312 & 432 & 656 \\
16 & 48 & 276 & 1,068 \\
18 & 448 & 3,440 & 15,066 \\
20 & 38,856 & 59,426 & 159,960 \\
\bottomrule
\end{tabular}
\end{table}

At both 36 and 48 terms, the search found 24 distinct orders, identifying reversals, in every one of the 60 instance/topology combinations. Since every search reached its candidate cap, these counts give lower bounds on the number of optima. The physical variation among optima is already visible at small sizes: exhaustive evaluation of the three 10-term minimum-cost sets gave routed-depth ranges of 58--83, 61--112, and 74--123 after averaging routing seeds. This is the variation that Stage~2 seeks to exploit.

\section{Selecting an Order Within the Plateau}
\label{sec:method}

\subsection{Candidate generation and depth selection}

Stage~1 minimizes $H_{\rm closed}$ using OR-Tools. On the 36- and 48-term instances, attaining $2n+2$ certifies optimality directly. Stage~2 explores only orders with that same cost. It proposes segment reversals or single-occurrence relocations, accepts only zero-cost changes, and retains distinct candidates after identifying reversals. The proposal cap is 100,000 and the candidate cap is $K=24$. The algorithm stops without a selection if it cannot fill the pool at the required cost. The selected order minimizes routed depth within this pool under a specified routing configuration and seed. This is a classical per-instance search, with no learned policy transferred between instances. It guarantees preservation of the primary cost. It guarantees neither a globally minimum routed depth nor a reduction in routed two-qubit count.

\begin{algorithm}[t]
\algorithmcaption{Strict selection among minimum-support-cost orders.}
\label{alg:plateau}
\centering
\begin{minipage}{0.98\columnwidth}\small
\hrule\medskip
\textbf{Input:} A certified minimum-cost order; fixed placement and lowering; routing seed; $K=24$.\\
\textbf{Output:} A valid order with unchanged $H_{\rm closed}$.
\begin{enumerate}\setlength{\itemsep}{2pt}\setlength{\parskip}{0pt}
\item Initialize a pool with the supplied order, identifying reversal-equivalent orders.
\item For at most 100,000 proposals, choose two distinct positions uniformly. With equal probability, reverse the segment or relocate the lower-position occurrence to the higher position.
\item Accept only proposals with zero change in $H_{\rm closed}$; retain new orders until the pool contains $K$ members.
\item Reject if fewer than $K$ candidates were found.
\item Route each candidate under the specified seed. Return the minimum-depth order, breaking ties by occurrence tuple.
\end{enumerate}
\medskip\hrule
\end{minipage}
\end{algorithm}

We compare selection with the \emph{reference order}: the support-optimal order used to initialize the candidate search. To identify an order with its reversal, we store the lexicographically smaller of their occurrence-identifier tuples. The reference occupies candidate index zero and remains fixed as the search adds orders. Its choice uses no routing results. For the 36- and 48-term assignment benchmarks, the reference is the lexicographically smallest minimum-cost order found in a preliminary bounded search spanning the optimum and the next two support-cost levels. It therefore depends on the search that produced those starting orders; it need not be the smallest order in the subsequent 24-candidate pool. For MaxCut and Pittsburgh, the reference is the support-optimal order returned by OR-Tools, with the same reversal rule.

The comparison measures the benefit of selecting among equal-cost alternatives to this starting order. The assignment benchmark reference includes the work of its preliminary search. Stage~2 then adds candidate generation and routing, whose costs we report separately. We also compare against the orders returned by two Stage-1 formulations. Identifying reversals removes duplicate logical orders from the search, although an order and its reversal can still route differently.

\subsection{Evaluation across routing seeds}

Selecting and evaluating on the same routed circuits would favor candidates partly because of routing luck. We therefore select with SABRE seed 1210003 and evaluate the chosen order with seed 1210019, then reverse the two roles. Each comparison uses the same instance, topology, and placement. This tests transfer across two routing seeds of the same router; transfer across instances or routers is a separate question. For instance $s$, topology $t$, and routing direction $r$, write the evaluated depth difference as $\delta_{str}=D_{str}^{\rm selected}-D_{str}^{\rm comparator}$. The term-seed observation is
\begin{equation}
\Delta_s=\frac{1}{6}\sum_{t=1}^{3}\sum_{r=1}^{2}\delta_{str}.
\label{eq:aggregation}
\end{equation}
In each direction, selection uses circuits compiled with one routing seed; the depth difference is measured on circuits compiled with the other seed. Reported mean differences average the 20 values $\Delta_s$; percentages use the difference of unrounded aggregate means divided by the comparator mean. The six routing conditions and 24 candidates contribute to each instance estimate; the sample size for inference is 20.

Before direct-depth selection, we used a selector that minimized $L_2$, the number of two-qubit operations on the heaviest dependency path. It assigns weight one to two-qubit gates and zero to other operations while retaining all dependencies, then takes the maximum weighted path length. In the first evaluation of that selector, both routing seeds contributed to selection and evaluation; we corrected this overlap by selecting with one seed and evaluating with the other, as above. We specified direct-depth selection and the restart control before their analyses, after the earlier results were known. The results below report both selectors and the effect of this correction.

\section{Routing Evaluation}
\label{sec:evaluation}

\subsection{Workloads and comparisons}

The primary workload contains 36 distinct quadratic Ising terms on 16 logical qubits, derived from a synthetic four-chiplet, four-site assignment QUBO. A binary variable $x_{cs}$ represents assigning chiplet $c$ to site $s$. Seeded connectivity and power determine pair costs, while row and column penalties enforce one chiplet per site. Substituting $x=(1-Z)/2$ gives a pool of quadratic $Z_iZ_j$ terms. For each term seed, we sample 36 distinct terms without replacement with probability proportional to coefficient magnitude, then sort by source-pool index. The 48-term arm follows the same rule. The maximum distinct-pair pool is ${16\choose2}=120$.

We use 20 term seeds and three fixed coupling graphs. Heavy hex is the distance-5 graph with 57 vertices, 64 edges, and diameter 16. Each modular graph joins two distance-3 modules with 38 total vertices; one bridge gives 41 edges and diameter 17, while three bridges give 43 edges and diameter 11. Bridge endpoints are the lowest-degree local vertices with identifiers breaking ties. We average over these three fixed graphs. Placement is deterministic and shared by all methods for each instance and topology. Logical qubits are ranked by descending occurrence frequency, then descending total absolute phase coefficient, then identifier. Physical vertices are ranked by descending degree, then descending closeness centrality, then identifier. Rank pairing assigns the 16 logical wires, and the highest-ranked unused vertex is assigned to the ancilla. Each pair support has a fixed deterministic shortest-path realization. Changing term order therefore changes neither the input placement nor the rule for realizing a support.

OR-Tools 9.15.6755 receives a five-second budget, using a cheapest-arc initial tour followed by guided local search; returned orders are rescored with the support cost in Eq.~\ref{eq:hclosed}. Qiskit 2.4.1 uses optimization level 3, the basis \texttt{rz,sx,x,cx}, fixed initial layout, and SABRE routing. The routing seeds are those defined in Section~\ref{sec:method}. Synthetic routed two-qubit counts are CX gates, whereas the hardware study uses native two-qubit operations. Exact graph lists, seed lists, and remaining software defaults accompany the reproducibility artifact.

The 48-term replication retains the same 20 term seeds, three topologies, two routing seeds, and $K=24$, at certified support cost 98. The second generator draws 36 distinct unweighted edges from the 120 possible edges of a 16-vertex random MaxCut graph, rejecting disconnected graphs through the deterministic random stream. Its coefficients are one and it has 20 graph seeds. It retains the placement rule, compiler settings, strict selection, and term-seed aggregation. A separate cross-router control evaluates the existing SABRE-selected orders with BasicSwap, holding placement, basis, optimization level, and the candidate pool fixed. These experiments test size, workload generator, and router separately.

The statistical unit is the term seed under Eq.~\ref{eq:aggregation}. We report paired 95\% percentile-bootstrap intervals from 10,000 resamples and exact two-sided Wilcoxon signed-rank tests~\cite{wilcoxon1945,efron1979}. Zero differences are dropped, tied absolute differences receive average ranks, and the complete sign distribution is enumerated. The minimum two-sided $p$-value for 20 nonzero differences is $2/2^{20}=1.91\times10^{-6}$, so unanimous comparisons share a finite testing floor. Intervals are pointwise and percentages use unrounded means; multiplicity-adjusted decisions use the reported $p$-values. Depth, routed CX, and two-qubit depth form a separate three-test Holm family for each of the corrected $L_2$, direct-depth, restart, 48-term, MaxCut, and BasicSwap protocols~\cite{holm1979}. The adjacent-stratum experiment uses one six-test family (two strata times three endpoints); the topology sensitivity uses six tests across Stage-1 CX and Stage-2 depth on three graphs. The Pittsburgh hardware study has one raw-error primary test and a two-test secondary family. We also report a global Holm sensitivity analysis over 64 comparisons in Section~\ref{sec:multiplicity}.

The earlier surrogate $\JSPT=H_{\rm closed}+0.05T$, where $T$ counts mapped-tree edge changes, serves as a Stage-1 comparator; it gave no significant routed-CX advantage over support-only ordering, and SM Sections~\ref*{sm-s:objectives}--\ref*{sm-s:stage1} give its definition and earlier comparisons, including the Default comparison added after results were available.

\subsection{Equal primary cost, lower routed depth}
\label{sec:stage2selection}

Direct-depth selection reduced mean opposite-seed depth from 228.6 to 199.3, a 12.83\% reduction, with lower depth in all 20 instances. The mean paired difference was $-29.3$ layers, with interval $[-34.0,-24.7]$ and Holm-adjusted $p=5.72\times10^{-6}$. Two-qubit depth also fell; the routed-CX confidence interval included zero (Table~\ref{tab:plateaudepth}). The $L_2$ selector, evaluated with the separate routing seed, reduced mean opposite-seed depth from 228.6 to 203.1 (11.14\%), also with lower depth in all 20 instances. Using the same two routed versions for both selection and evaluation had given a 12.51\% reduction; separating their roles reduced the estimate by 1.37 percentage points. Direct depth is lower still: the direct-minus-$L_2$ comparison gives $-3.9$ layers, interval $[-5.6,-2.1]$, and 15/2/3 favorable/tied/adverse term seeds. Table~\ref{tab:plateaudepth} shows the tradeoff: $L_2$ yields fewer CX gates on average, while direct depth is the better depth selector in this pool.

\begin{table*}[t]
\caption{Opposite-routing-seed controls at 36 terms, transfer at 48 terms and on MaxCut, and cross-router evaluation. Negative differences favor selection. At 36 terms, $L_2$ and direct selection use the reference order (228.6 depth); the restart arm uses its opposite-fold mean (229.3). The BasicSwap rows evaluate SABRE-selected direct-depth candidates against the same reference order. ``Diff.'' is computed from unrounded aggregate means; displayed means are rounded. The six row groups belong to separate three-endpoint Holm families.}
\label{tab:plateaudepth}
\centering\footnotesize
\setlength{\tabcolsep}{7pt}
\begin{tabular}{lrrrrr}
\toprule
Selector / metric & Selected & Comp. & Diff. & 95\% CI & Holm $p$ \\
\midrule
$L_2$: depth & 203.1 & 228.6 & $-25.5$ & $[-30.0,-21.1]$ & $5.72\!\times\!10^{-6}$ \\
$L_2$: CX & 415.6 & 426.1 & $-10.4$ & $[-19.0,-1.1]$ & 0.0362 \\
$L_2$: 2Q depth & 174.2 & 195.4 & $-21.2$ & $[-25.3,-17.2]$ & $5.72\!\times\!10^{-6}$ \\
Direct: depth & 199.3 & 228.6 & $-29.3$ & $[-34.0,-24.7]$ & $5.72\!\times\!10^{-6}$ \\
Direct: CX & 420.2 & 426.1 & $-5.9$ & $[-14.9,4.1]$ & 0.2305 \\
Direct: 2Q depth & 176.3 & 195.4 & $-19.1$ & $[-23.1,-14.9]$ & $7.63\!\times\!10^{-6}$ \\
Restart: depth & 228.4 & 229.3 & $-0.9$ & $[-2.5,0.5]$ & 0.5035 \\
Restart: CX & 426.0 & 425.7 & $+0.3$ & $[-0.5,1.3]$ & 1.0000 \\
Restart: 2Q depth & 195.2 & 195.6 & $-0.4$ & $[-1.9,0.8]$ & 1.0000 \\
\midrule
48-term direct: depth & 258.9 & 287.1 & $-28.2$ & $[-34.5,-22.1]$ & $5.72\!\times\!10^{-6}$ \\
48-term direct: CX & 518.4 & 520.8 & $-2.4$ & $[-8.9,3.8]$ & 0.7012 \\
48-term direct: 2Q depth & 225.4 & 243.6 & $-18.2$ & $[-23.4,-13.1]$ & $1.14\!\times\!10^{-5}$ \\
\midrule
MaxCut direct: depth & 200.4 & 244.2 & $-43.8$ & $[-50.3,-37.1]$ & $5.72\!\times\!10^{-6}$ \\
MaxCut direct: CX & 417.6 & 448.5 & $-31.0$ & $[-39.9,-22.2]$ & $5.72\!\times\!10^{-6}$ \\
MaxCut direct: 2Q depth & 176.9 & 210.5 & $-33.6$ & $[-39.5,-27.4]$ & $5.72\!\times\!10^{-6}$ \\
\midrule
BasicSwap direct: depth & 278.8 & 273.9 & $+5.0$ & $[1.7,8.3]$ & 0.0105 \\
BasicSwap direct: CX & 611.6 & 596.6 & $+15.0$ & $[7.8,22.2]$ & 0.00445 \\
BasicSwap direct: 2Q depth & 272.6 & 267.0 & $+5.6$ & $[2.2,9.0]$ & 0.00938 \\
\bottomrule
\end{tabular}
\end{table*}

The choice of comparison order changes the reported magnitude. Table~\ref{tab:baseline-reconcile} names all three orders behind the motivating example. Gui-open OR-Tools minimizes the established open support-path cost, which differs by a constant four here; $\JSPT$ OR-Tools adds the mapped-tree tie refinement defined above. Both attain support cost 74. Comparing the same selected circuits against them gives 14.77\% and 22.39\% depth reductions. We use comparison with the reference order for the 12.83\% headline: the larger percentages change the comparator, not the selected circuit or the experiment. The two tables use separate bootstrap resamples, producing slightly different interval endpoints for the same mean effect.

\begin{table*}[t]
\caption{Stage-1 tie-break sensitivity on the common 36-term units. Every baseline has $H_{\rm closed}=74$. The last columns compare the same opposite-routing-seed direct-depth Stage-2 result (mean 199.3) with each baseline; negative favors Stage~2. Intervals use the baseline-sensitivity resamples; Table~\ref{tab:plateaudepth} gives the separate direct-selector analysis.}
\label{tab:baseline-reconcile}
\centering\footnotesize
\begin{tabular}{lrrrrrr}
\toprule
Baseline / role & $H_{\rm closed}$ & Mean CX & Mean depth & $\Delta$ depth & 95\% CI & Change \\
\midrule
Reference order & 74 & 426.1 & 228.6 & $-29.3$ & $[-34.2,-24.7]$ & $-12.83\%$ \\
Gui-open OR-Tools & 74 & 431.9 & 233.8 & $-34.5$ & $[-43.1,-26.4]$ & $-14.77\%$ \\
$\JSPT$ OR-Tools & 74 & 430.2 & 256.7 & $-57.5$ & $[-66.3,-49.5]$ & $-22.39\%$ \\
\bottomrule
\end{tabular}
\end{table*}

\begin{figure*}[t]
\centering
\includegraphics[width=\textwidth]{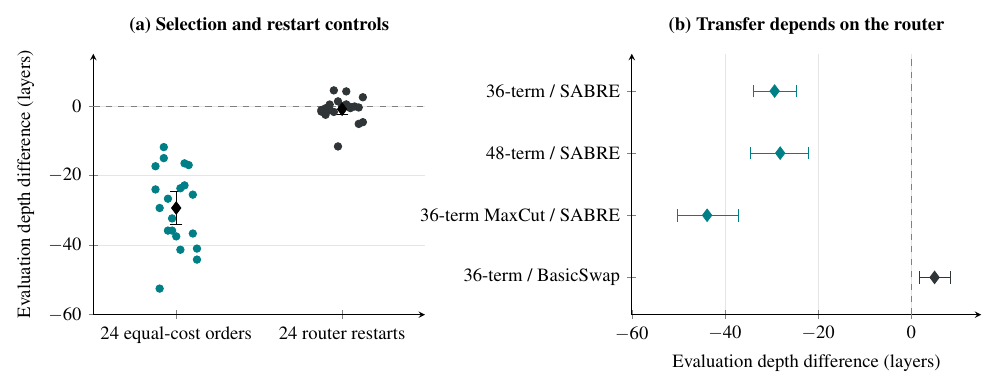}
\caption{Selection controls and transfer across compilation settings. (a) Each point is one term-seed depth difference at 36 terms. Direct selection compares with the reference order; the matched restart control uses its opposite-fold comparator. Diamonds and bars show means and paired 95\% bootstrap intervals. (b) Direct-depth selection transfers across size and generator under SABRE, but reverses under BasicSwap. Intervals are pointwise; the family-adjusted tests and all three compiler metrics appear in Table~\ref{tab:plateaudepth}. Negative differences favor selection in both panels.}
\label{fig:selection}
\end{figure*}

\subsection{Selection controls and routing variability}

The matched restart control spends the same 24-route budget recompiling the reference order with different transpiler seeds. Its 24 seeds are split into two ordered folds of 12. The minimum-depth position in one fold selects the same position in the other fold, where its depth is compared with that fold's mean; the fold roles are then reversed. This tests whether a restart position selected under one set of seeds remains favorable under a different set. Its reference depth is the fold mean, 229.3; the candidate experiment uses the two fixed SABRE seeds and has reference depth 228.6. Selecting a restart changed opposite-fold depth by $-0.9$ layers ($-0.41\%$), interval $[-2.5,0.5]$, Holm $p=0.5035$ (Fig.~\ref{fig:selection}a). Keeping the best compiled circuit from all 24 restarts instead gave mean depth 222.0 versus 229.3 ($-3.21\%$) in the in-sample analysis. Thus restarts help when the compiler keeps the winning circuit, while the cross-fold test gives no evidence that the winning position predicts lower depth under independent seeds.

Candidate diversity alone was also insufficient. A routing-independent control selected the minimum SHA-256 hash of each candidate's occurrence tuple after identifying reversals. It changed mean depth by $+1.44\%$, interval $[-3.7,10.7]$ layers, $p=0.641$. The corrected $L_2$ selector reduced depth by 13.06\% relative to the exact mean of all 24 candidates and by 12.40\% relative to the hash choice. These comparisons use the unchanged candidate pool and ask whether its routed ranking adds information beyond obtaining a different legal order. The score-permutation control breaks that information directly. Across 1,000 within-pool permutations of training $L_2$ scores among candidate identities, the mean null depth difference was $+4.6$ layers, with empirical 2.5--97.5\% range $[0.8,8.7]$. None matched the actual $-25.5$-layer difference (plus-one $p=0.0010$). The range describes the randomized outcomes. These hash, uniform, and permuted-score controls evaluate the corrected $L_2$ selector.

Candidate rankings are strongly correlated across routing seeds. Within each 24-candidate pool, depths under the two SABRE seeds had median Spearman correlation 0.910 (IQR 0.810--0.945). For pools compiled with the evaluation seed, the medians of the mean, minimum, and maximum depths were 228.7/191.0/277.5. The median sample SD was 22.5 layers (IQR 19.0--25.2), and the median range was 85.5 (IQR 75.0--94.0). Under the routing seed excluded from selection, the selected order's depth lay a median 1.57 pool SD below the pool mean (IQR 1.36--1.82), within the roughly $2\sigma$ scale of a naive normal best-of-24 comparison. The gain therefore has the scale expected from choosing the best of a small pool. The normal comparison is a reference scale, not a fitted distribution. What makes selection useful here is that a low-depth order under one routing seed tends to remain low-depth under the other, while every candidate preserves the primary optimum.

\subsection{Transfer across sizes, generators, and routers}

At 48 terms, direct selection reduced opposite-seed depth by 9.81\%; on the random-MaxCut generator it reduced depth by 17.95\% (Fig.~\ref{fig:selection}b). Both reduced depth in all 20 instances and remained significant after family and global correction. At 48 terms, the routed-CX interval included zero. Selection therefore extends to a larger phase component and a different pair-support generator. Depth also decreases on each of the three coupling graphs at 36 terms. Direct selection changed depth by $-21.7$ layers on heavy hex ($-10.27\%$, interval $[-32.1,-11.8]$), $-39.7$ on the one-bridge graph ($-15.17\%$, $[-48.8,-30.9]$), and $-26.6$ on the three-bridge graph ($-12.52\%$, $[-39.8,-13.6]$). Their six-test-family $p$-values were 0.00276, $1.14\times10^{-5}$, and 0.000839, respectively; all three also survive the global sensitivity. The overall decrease is therefore shared across all three topologies.

BasicSwap gives the complementary result. Orders selected using SABRE depth increased BasicSwap depth by 1.82\%, routed CX by 2.51\%, and two-qubit depth by 2.08\%. All three increases were significant within their three-test family; their global-adjusted $p$-values exceed 0.05 (Section~\ref{sec:multiplicity}). The median SABRE--BasicSwap depth-rank correlation was only 0.093. A useful order for one routing heuristic can therefore be an unfavorable choice for another; the method should be assessed with its intended router.

\subsection{Comparison with higher-cost strata}

The adjacent-stratum control tests whether a small primary-cost penalty buys better routed depth. It generates separate 24-candidate pools at exactly $H_{\min}+2$ and exactly $H_{\min}+4$, retaining the same proposal budget, routing settings, direct-depth selector, and opposite-seed evaluation. It never pools those strata. Their selected depths were 206.9 and 209.3, compared with 199.3 for strict selection (Table~\ref{tab:adjacentstrata}). The six-test-adjusted depth $p$-values were 0.0592 and 0.00946. Global correction leaves both CX penalties significant, but neither depth penalty (Section~\ref{sec:multiplicity}). Selection remains useful inside the relaxed strata. Relative to each pool's own routing-independent first candidate, depth fell from 241.3 to 206.9 at $H_{\min}+2$ (14.28\%) and from 246.9 to 209.3 at $H_{\min}+4$ (15.24\%), with lower depth in all 20 instances in each descriptive comparison. Selection thus helps within both relaxed strata, yet strict selection achieves the lowest mean depth and preserves the support optimum. Each percentage uses its own pool reference; comparing a relaxed selected mean with the strict-pool reference mean would mix the effect of selection with the change in pool.

\begin{table*}[t]
\caption{Effect of relaxing the support optimum at fixed $K=24$. Differences are adjacent-stratum selection minus strict $H_{\min}$ selection; positive is worse. Displayed means are rounded, while differences and percentages use unrounded values. Family size is six.}
\label{tab:adjacentstrata}
\centering\footnotesize
\setlength{\tabcolsep}{7pt}
\begin{tabular}{lrrrrr}
\toprule
Stratum / metric & Adjacent & Strict & Diff. & 95\% CI & Holm $p$ \\
\midrule
$H_{\min}+2$: depth & 206.9 & 199.3 & $+7.6$ & $[1.6,14.6]$ & 0.0592 \\
$H_{\min}+2$: CX & 436.0 & 420.2 & $+15.9$ & $[8.6,23.1]$ & 0.00241 \\
$H_{\min}+2$: 2Q depth & 182.5 & 176.3 & $+6.1$ & $[0.7,12.1]$ & 0.0592 \\
$H_{\min}+4$: depth & 209.3 & 199.3 & $+10.0$ & $[4.1,16.6]$ & 0.00946 \\
$H_{\min}+4$: CX & 446.1 & 420.2 & $+26.0$ & $[18.0,34.5]$ & $3.43\!\times\!10^{-5}$ \\
$H_{\min}+4$: 2Q depth & 186.6 & 176.3 & $+10.3$ & $[5.1,15.5]$ & 0.00676 \\
\bottomrule
\end{tabular}
\end{table*}

\subsection{Compilation cost and candidate budget}

Candidate generation took a median 0.0265~s per instance/topology combination, while routing cost a median 0.0131~s per candidate route and approximately 0.646~s for 24 candidates under both seeds. Final selection took about $8\,\mu$s. These timings describe the experimental workload and computing environment. Stage~1 uses a five-second OR-Tools budget. The preliminary reference-order search for the assignment benchmarks is excluded from these additional Stage-2 timings. With direct support rescoring, candidate generation costs $O(B_{\rm walk}n)$ for proposal budget $B_{\rm walk}$; the main additional empirical expense is compiling the candidates. Figure~\ref{fig:stage2headline} makes the budget tradeoff visible. Nested hash-ordered subsets use $K=1,2,4,8,12,16,24$, to show how performance changes with candidate count. Direct-depth changes were $+1.44$, $-2.46$, $-6.81$, $-9.80$, $-10.80$, $-11.63$, and $-12.83\%$. The final 1.21-percentage-point improvement exceeds the specified one-point saturation threshold, so saturation is not established at 24. The $L_2$ curve improved only 0.45 percentage points from 16 to 24, but its behavior cannot establish saturation for the different direct-depth selector. Larger pools may therefore offer further depth reductions at additional compilation cost.

\subsection{Global multiple-testing correction}
\label{sec:multiplicity}

Global Holm correction over 64 comparisons leaves 28 significant at 0.05. Table~\ref{tbl:global_audit} shows the results relevant to plateau selection. Depth reductions at 36 terms, 48 terms, and on MaxCut each have global $p=1.22\times10^{-4}$, and their two-qubit-depth reductions also remain significant. The corresponding 36- and 48-term CX differences and the restart contrast have $p>0.05$. The main SABRE depth result therefore holds under this broader correction. The broader correction changes the interpretation of several secondary results. BasicSwap's adjusted $p$-values are 0.3051 for depth, 0.0534 for CX, and 0.1455 for two-qubit depth. The adjacent-stratum depth penalties also exceed 0.05 ($p=0.7985$ and $0.1009$), while both CX penalties remain significant ($p=0.0183$ and $2.69\times10^{-4}$). All three Pittsburgh hardware endpoints have global $p=1$. The statistical evidence for lower depth under SABRE is consequently stronger than the evidence for higher depth under BasicSwap.

This sensitivity analysis includes direct-depth selection, restarts, transfer, adjacent strata, topology, and Pittsburgh hardware, together with the Stage-1 ordering, ParitySynth resynthesis, Kingston, and Boston comparisons. Its 64-test set excludes the corrected $L_2$ family, the descriptive direct-versus-$L_2$ contrast, randomized controls, and the conditional shot-noise analysis. Those analyses retain their separately reported roles. SM Section~\ref*{sm-s:audit} lists all comparison identities and the earlier versions of the correction.

\begin{table*}[t]
\caption{Global multiple-testing sensitivity for the main results. Rows show raw $p$-values, adjustment within the study's family, and adjustment across the 64-comparison set. The corrected $L_2$ family, direct-versus-$L_2$ descriptive comparison, randomized controls, and fixed-panel shot analysis are outside it.}
\label{tbl:global_audit}
\centering\footnotesize
\setlength{\tabcolsep}{12pt}
\begin{tabular}{lrrr}
\toprule
Comparison / endpoint & Raw $p$ & Family $p$ & Global-64 $p$\\
\midrule
36-term direct depth & $1.91\!\times\!10^{-6}$ & $5.72\!\times\!10^{-6}$ & $1.22\!\times\!10^{-4}$ \\
36-term direct 2Q depth & $3.81\!\times\!10^{-6}$ & $7.63\!\times\!10^{-6}$ & $1.83\!\times\!10^{-4}$ \\
Restart depth & $0.1678$ & $0.5035$ & $1.0000$ \\
48-term direct depth & $1.91\!\times\!10^{-6}$ & $5.72\!\times\!10^{-6}$ & $1.22\!\times\!10^{-4}$ \\
MaxCut direct depth & $1.91\!\times\!10^{-6}$ & $5.72\!\times\!10^{-6}$ & $1.22\!\times\!10^{-4}$ \\
$H_{\min}+2$ depth penalty & $0.0296$ & $0.0592$ & $0.7985$ \\
$H_{\min}+2$ CX penalty & $4.83\!\times\!10^{-4}$ & $0.0024$ & $0.0183$ \\
$H_{\min}+4$ depth penalty & $0.0032$ & $0.0095$ & $0.1009$ \\
$H_{\min}+4$ CX penalty & $5.72\!\times\!10^{-6}$ & $3.43\!\times\!10^{-5}$ & $2.69\!\times\!10^{-4}$ \\
BasicSwap depth increase & $0.0105$ & $0.0105$ & $0.3051$ \\
BasicSwap CX increase & $0.0015$ & $0.0044$ & $0.0534$ \\
BasicSwap 2Q depth increase & $0.0047$ & $0.0094$ & $0.1455$ \\
Hardware raw generator error & $0.2722$ & $0.2722$ & $1.0000$ \\
Hardware mitigated error & $0.2948$ & $0.5443$ & $1.0000$ \\
Hardware signed expectation & $0.2722$ & $0.5443$ & $1.0000$ \\
\bottomrule
\end{tabular}
\end{table*}

\begin{figure}[t]
\centering
\includegraphics[width=\columnwidth]{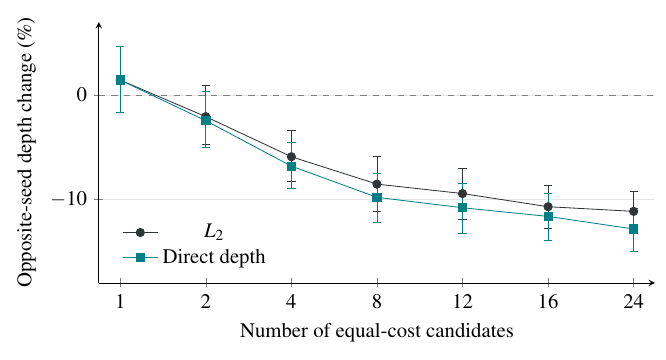}
\caption{Depth reduction versus candidate budget at 36 terms, relative to the reference order. Budgets use nested hash-ordered candidate subsets; their spacing is categorical. Points give aggregate percent changes across 20 term-seed units. Direct-depth bars are paired bootstrap intervals for percent change; $L_2$ bars divide the paired depth-difference interval by the fixed comparator mean. They are descriptive pointwise intervals, not simultaneous confidence bands. The direct-depth gain from 16 to 24 candidates exceeds the specified one-point saturation threshold.}
\label{fig:stage2headline}
\end{figure}

\section{Hardware Evaluation}
\label{sec:hardware}

\subsection{Pittsburgh study protocol}

We evaluated the reference order and the direct-depth-selected order on \texttt{ibm\_pittsburgh}, an IBM Heron processor, using 40 new 36-term seeds. We recorded the protocol before selecting the backend, compiling the circuits, or executing them. Selection and execution used opposite routing seeds in both directions, with two repeated blocks. To make the ideal output efficiently verifiable, the probe replaces phase angles with signed Clifford angles $\operatorname{sign}(c_i)\pi/2$. Supports, occurrence identities, order, placement, and routed two-qubit structure are retained. Starting from $\lvert+\rangle^{16}\lvert0\rangle$, we measure the 16 signed data-wire stabilizer generators and validate the ancilla's $Z$ stabilizer and return to $\lvert0\rangle$. These Clifford angles make the probe classically verifiable while preserving the ordering choice; they replace the application angles determined by the Ising coefficients. The primary endpoint is mean raw generator error $(1-s_gE_g)/2$, where $s_g$ is the ideal sign and $E_g$ the measured expectation. Zero means agreement with a stabilizer generator and one means the opposite sign. Their average measures error for the specified probe state. Results are averaged over the 16 generators, two routing directions, and two execution blocks before paired inference across 40 term seeds. Shot outcomes and compatible generators within a measurement setting are not independent term-seed observations.

The main experiment used 1,440 measurement circuits at 1,024 shots and 640 calibration circuits at 4,096 shots, totaling 2,080 circuits and 4,096,000 shots. A preliminary validation used 3,328 shots and identified missing input-state preparation. We corrected the preparation and repeated the check before the main experiment, keeping the instance seeds, endpoint, budget, and analysis plan unchanged. Readout-mitigated generator error and raw signed stabilizer expectation form the secondary two-test Holm family. Mitigation uses independent one-qubit assignment maps from zero/one calibration circuits; corrected expectations are clipped to $[-1,1]$ and a calibration is rejected when its affine slope has magnitude below 0.05. This procedure corrects readout error; gate errors remain part of the measured endpoint.

\subsection{Error reduction and variation across instances}

Stage~2 reduced mean raw generator error from \StageTwoQPURawStageOne{} to \StageTwoQPURawStageTwo{}, a change of \StageTwoQPURawDelta{} (\StageTwoQPURawPct{}). Error was lower in 24 of 40 instances. The primary 95\% confidence interval, \StageTwoQPUCI{}, included zero ($p=\StageTwoQPUP{}$), so the experiment does not establish a reduction across instances. The secondary mitigated-error and signed-expectation intervals also included zero (Table~\ref{tab:stage2qpu}). We also analyzed shot uncertainty conditional on the exact executed circuits. This analysis was added after the primary result and holds instances, routing directions, and blocks fixed while preserving covariance among generators measured in the same setting. Its standard error is \StageTwoQPUFixedSE{} and its interval \StageTwoQPUFixedCI{} excludes zero ($p=\StageTwoQPUFixedP{}$). The observed panel reduction is therefore larger than shot noise alone would explain. Table~\ref{tab:stage2qpu} places this result beside the primary interval: shot precision on these circuits is high enough to detect the small change, while variation across instances limits the population inference.

\begin{table*}[t]
\caption{Stage-2 execution on \texttt{ibm\_pittsburgh}. Differences are the selected order minus the reference order. W/T/L counts the direction favorable to Stage~2. The fixed-panel row propagates shot noise conditional on the exact executed circuits; all other rows use 40 term seeds as the scientific unit. Secondary $p$-values are Holm adjusted.}
\label{tab:stage2qpu}
\centering\footnotesize
\setlength{\tabcolsep}{3.2pt}
\begin{tabular}{llrrrrrr}
\toprule
Endpoint & Inference & Stage~2 & Stage~1 & Difference & W/T/L & 95\% CI & $p$ \\
\midrule
Raw generator error & term-seed primary & \StageTwoQPURawStageTwo{} & \StageTwoQPURawStageOne{} & \StageTwoQPURawDelta{} & \StageTwoQPUWTL{} & \StageTwoQPUCI{} & \StageTwoQPUP{} \\
Raw generator error & fixed-panel conditional & \StageTwoQPURawStageTwo{} & \StageTwoQPURawStageOne{} & \StageTwoQPURawDelta{} & --- & \StageTwoQPUFixedCI{} & $\StageTwoQPUFixedP$ \\
Mitigated generator error & term-seed secondary & \StageTwoQPUMitStageTwo{} & \StageTwoQPUMitStageOne{} & \StageTwoQPUMitDelta{} & \StageTwoQPUMitWTL{} & \StageTwoQPUMitCI{} & \StageTwoQPUMitHolmP{} \\
Signed stabilizer expectation & term-seed secondary & \StageTwoQPUSignedStageTwo{} & \StageTwoQPUSignedStageOne{} & \StageTwoQPUSignedDelta{} & \StageTwoQPUSignedWTL{} & \StageTwoQPUSignedCI{} & \StageTwoQPUSignedHolmP{} \\
\bottomrule
\end{tabular}
\end{table*}

The physical circuits retained a depth reduction, averaging 34.3 layers across the 80 seed/direction cells, with 56 favorable, 15 tied, and nine adverse cells. Two-qubit depth changed by $-12.0$ (52/17/11), while native two-qubit count changed by only $-0.5$ (29/15/36). These descriptive counts share 40 instances. The hardware result illustrates the difference between a compiler endpoint and execution accuracy: a clear mean depth gain can coexist with a small error change and substantial instance variation.

Earlier hardware experiments explain our choice of the stabilizer probe. The six-seed Fez pilot developed a different protocol. Kingston tested $\JH=\JW+0.05T$, which adds mapped-tree edge changes $T$ to the support-and-topology score $\JW$ from our previous study, and found a mitigated-error contrast between its OR-Tools and stochastic methods (family $p=0.0438$, global $p=0.6131$). Boston directly tested $\JSPT$, not strict Stage~2: its OR-Tools-minus-Default mean absolute observable-error difference was $-0.0059$, interval $[-0.0155,0.0028]$, family $p=0.4917$. Its ideal-observable panel was dominated by small expectations: 86.5\% had magnitude below 0.25 and a zero predictor had MAE 0.0941, compared with hardware means 0.1009--0.1077. The stabilizer probe instead supplies known signed-generator targets. The earlier protocols and full results are reported separately in SM Section~\ref*{sm-s:legacyhardware}.

\section{Relation to Existing Compiler Optimizations}
\label{sec:related}

Cancellation-oriented ordering already treats Pauli-term sequencing as a traveling salesperson problem~\cite{tomesh2021ordering,gui2020extended}. In our all-$Z$ pair-support setting, the Gui transition reduces to support-mask Hamming distance, and its open-path cost differs from $H_{\rm closed}$ by the constant four. Neither that cost, the artificial-depot conversion, nor parity reuse through an ancilla is new. Our focus is the remaining minimum-cost set: its structural characterization and a strict secondary choice evaluated under independent routing randomness.

Parity-network and phase-polynomial synthesis can change the underlying realization more broadly~\cite{amy2018cnotphase,vandaele2022phasepolynomials,meijer2023architecture,cao2025graphic}. Paulihedral, Rustiq, MonteQ, and 2QAN combine ordering with block, synthesis, routing, or scheduling freedoms~\cite{li2022paulihedral,goubault2024rustiq,machiya2026monteq,lao2022twoqan}. Fixed placement and lowering isolate the choice studied here; they are not asserted to be the best full-compiler design. Our plateau experiments compare methods with the same placement and lowering rules. Our earlier ParitySynth comparison illustrates the effect of allowing broader synthesis freedom. It resynthesizes the same data-only phase polynomial using the full physical graph, including additional qubits for parity construction, whereas the restricted method preserves the 17-wire maintained-parity realization. In its separate 720-route comparison, restricted LKH-3 optimized $\JH$ and averaged \ParityLKH{} routed CX versus \ParitySynthCX{} for ParitySynth, but median compile time favored ParitySynth (3.14 versus 9.01~s) and the depth-difference interval crossed zero. After correcting an invalid circuit export, all routes passed phase-polynomial and final-linear-map checks; the artifact preserves both export versions. That comparison predates Stage~2 and uses different synthesis freedoms. Its sensitivity to the available routing qubits and topology is documented in SM Section~\ref*{sm-s:resynthesis}.

Neutral local search and compiler autotuning also explore equal-objective or semantically equivalent alternatives~\cite{prestwich2005symmetry,triantafyllis2003ose,ashouri2018autotuning}. The present method applies that general idea within a certified quantum support-cost optimum. SABRE supplies the downstream routing heuristic~\cite{li2019sabre}; candidate generation does not learn across instances. Finally, circuit depth need not rank scheduled runtime when gate durations differ~\cite{tremba2025depth}. Our synthetic depth result is therefore reported separately from native-device error and the Boston scheduled-duration measurements.

\section{Discussion and Conclusion}
\label{sec:discussion}

A minimum primary cost can leave a compiler with many physically different choices. For distinct pair supports, the line-graph characterization identifies those choices whenever the support lower bound is attained. The experiments show how to use them: generate several minimum-cost orders, compile them with the intended router, and select by the downstream quantity of interest. The strict constraint makes the tradeoff explicit---additional classical compilation work in exchange for lower routed depth, with no increase in logical support-transition cost. For a compiler already using this maintained-parity implementation, the evidence supports a concrete decision sequence. First certify the support optimum if the lower bound is attainable. Then consider secondary selection when distinct optima can be found and the measured extra routing cost is acceptable. Evaluate with the intended router and keep the validation routes separate from those used for selection. If the candidate search cannot fill the required strict pool, the present protocol stops; the experiments do not justify silently relaxing the optimum or assuming the same gain from fewer candidates.

The controls explain the gain. Primary-cost ties contain a broad distribution of routed depths, and candidate rankings are strongly correlated across the two SABRE seeds. Selecting from this distribution produces the expected scale of a best-of-24 improvement. The contribution is making that choice available while preserving a certified logical optimum, and measuring its benefit against explicit compilation costs and alternatives. The experiments also identify the remaining questions. Exact multiplicity counts stop at 20 terms; counts at 36 and 48 terms are lower bounds from the candidate search. Routing experiments cover 36 and 48 terms, two synthetic pair-support generators, three fixed topologies, and two SABRE seeds. Stage~2 samples 24 orders rather than optimizing depth over the full plateau. Joint placement, broader parity synthesis, and full application circuits are natural extensions.

The router comparison gives a practical rule: select with the compiler that will execute the circuit. Rankings persist across SABRE seeds but weaken sharply under BasicSwap. Hardware introduces a further requirement, since lower routed depth alone does not ensure lower execution error across instances. Within the tested SABRE pipeline, however, the result is consistent across sizes, generators, and topologies: choosing among equal-cost phase-term orders reduces depth without sacrificing the support optimum.

\section*{Data and Code Availability}

Code, reproducibility materials, recorded protocols, result tables, and public-safe execution metadata for this study are available at \url{https://github.com/owenfriedewald/compiler-freedom}; the repository release manifest identifies the evidence files and their integrity records. The versioned archival release is identified by DOI \href{https://doi.org/10.5281/zenodo.22063212}{10.5281/zenodo.22063212}. Software is released under Apache-2.0, and research data and artifacts under CC BY 4.0. Raw provider job references and service/account identifiers are excluded, while their hashes and the corresponding circuit, count, seed, backend, and calibration records are retained. The companion Supplementary Material gives the full protocols, secondary comparisons, and the separate Fez, Kingston, and Boston hardware studies.

\section*{Conflict of Interest}

The authors declare no conflicts of interest.

\section*{Acknowledgment}

The computation for this work was performed on the University of Missouri's Quantum Innovation Center, in partnership with IBM Quantum and facilitated by Research Support Services at the University of Missouri, Columbia, MO. The facility record is available at DOI: \href{https://doi.org/10.32469/10355/107781}{10.32469/10355/107781}.

OpenAI Codex assisted with grammatical revision and reorganization of the Abstract and Sections~I--VII, and was used more extensively in the drafting and revision of the Supplementary Material. It also assisted with figure generation, primarily the maintained-parity circuit illustration in Fig.~\ref{fig:maintained}. Codex was further used as a review tool to identify potential issues in mathematical arguments and to suggest refinements, including the proof in Section~\ref{sec:contract} and the weighted-refinement proof in SM Section~\ref*{sm-s:objectives}. The authors retain responsibility for the manuscript and its scientific claims.

\bibliographystyle{IEEEtran}
\bibliography{references}
\EOD
\end{document}

%% file: data/final_numbers.tex
\newcommand{\ParityLKH}{430.2}
\newcommand{\ParitySynthCX}{543.6}

\newcommand{\StageTwoQPURawStageTwo}{0.4130}
\newcommand{\StageTwoQPURawStageOne}{0.4155}
\newcommand{\StageTwoQPURawDelta}{-0.0025}
\newcommand{\StageTwoQPURawPct}{-0.59\%}
\newcommand{\StageTwoQPUWTL}{24/0/16}
\newcommand{\StageTwoQPUCI}{[-0.0142,\,0.0104]}
\newcommand{\StageTwoQPUP}{0.2722}
\newcommand{\StageTwoQPUMitStageTwo}{0.4095}
\newcommand{\StageTwoQPUMitStageOne}{0.4120}
\newcommand{\StageTwoQPUMitDelta}{-0.0025}
\newcommand{\StageTwoQPUMitCI}{[-0.0149,\,0.0110]}
\newcommand{\StageTwoQPUMitWTL}{25/0/15}
\newcommand{\StageTwoQPUMitHolmP}{0.5443}
\newcommand{\StageTwoQPUSignedStageTwo}{0.1739}
\newcommand{\StageTwoQPUSignedStageOne}{0.1690}
\newcommand{\StageTwoQPUSignedDelta}{+0.0049}
\newcommand{\StageTwoQPUSignedCI}{[-0.0208,\,0.0283]}
\newcommand{\StageTwoQPUSignedWTL}{24/0/16}
\newcommand{\StageTwoQPUSignedHolmP}{0.5443}
\newcommand{\StageTwoQPUFixedCI}{[-0.0034,\,-0.0016]}
\newcommand{\StageTwoQPUFixedP}{4.46\!\times\!10^{-8}}
\newcommand{\StageTwoQPUFixedSE}{0.000451}